\documentclass[sigplan,twocolumn,nonacm]{acmart}
\renewcommand\footnotetextcopyrightpermission[1]{}
\usepackage{cleveref}
\usepackage{graphicx}
\usepackage{subcaption}
\usepackage{xspace}
\usepackage{listings}
\usepackage{url}
\usepackage{amsmath}
\usepackage{cleveref}
\usepackage{algorithm}
\usepackage{algpseudocode}
\usepackage{booktabs}
\usepackage{pifont}
\usepackage{graphicx}
\usepackage{breqn}

\setcitestyle{nosort}
\lstdefinelanguage{Rust}{
  morekeywords={as, break, const, continue, crate, else, enum, extern, false, fn, for, if, impl, in, let, loop, match, mod, move, mut, pub, ref, return, Self, self, static, struct, super, trait, true, type, unsafe, use, where, while, async, await, dyn},
  morecomment=[l]{//},
  morecomment=[s]{/*}{*/},
  morestring=[b]",
  morestring=[b]',
}

\iffalse
\newcommand{\newcommenter}[3]{%
  \newcommand{#1}[1]{%
    \textcolor{#2}{\small\textsf{[{#3}: {##1}]}}%
  }%
}
\else
\newcommand{\newcommenter}[3]{\newcommand{#1}[1]{}}
\fi

\newcommenter{\lk}{blue}{LK}
\newcommenter{\yt}{orange}{YT}

\newcommand{\NormalSend}{\textit{Normal Send}\xspace}
\newcommand{\EagerSend}{\textit{Eager Send}\xspace}
\newcommand{\YCT}{\textit{YCT}\xspace}
\newcommand{\ACK}{\textit{ACK}\xspace}
\newcommand{\UNACK}{\textit{UNACK}\xspace}
\newcommand{\OB}{\textit{OB}\xspace}
\newcommand{\MODE}{\textit{MODE}\xspace}
\newcommand{\EAGERSENT}{\textit{EAGER\_SENT}} 

\begin{document}

\title{Can you keep a secret? A new protocol for sender-side enforcement of causal message delivery}

\author{Yan Tong}
\affiliation{%
  \institution{University of California}
  \city{Santa Cruz}
  \country{USA}}
\email{ytong24@ucsc.edu}

\author{Nathan Liittschwager}
\affiliation{%
  \institution{University of California}
  \city{Santa Cruz}
  \country{USA}}
\email{nliittsc@ucsc.edu}

\author{Lindsey Kuper}
\affiliation{%
  \institution{University of California}
  \city{Santa Cruz}
  \country{USA}}
\email{lkuper@ucsc.edu}

\sloppy

\begin{abstract}
Protocols for causal message delivery are widely used in distributed systems. Traditionally, causal delivery can be enforced either on the message sender's side or on the receiver's side. The traditional sender-side approach avoids the message metadata overhead of the receiver-side approach, but is more conservative than necessary. We present \emph{Cykas} (``Can you keep a secret?''), a new protocol for sender-side enforcement of causal delivery that sidesteps the conservativeness of the traditional sender-side approach by allowing eager sending of messages and constraining the behavior of their recipients. We implemented the Cykas protocol in Rust and checked the safety and liveness of our implementation using the Stateright implementation-level model checker. Our experiments show that for applications involving long-running jobs, Cykas has a performance advantage: Cykas lets long-running jobs start (and end) earlier, leading to shorter overall execution time compared to the traditional sender-side approach.  
\end{abstract}

\maketitle

\section{Introduction}\label{sec:introduction}

One of the challenges of implementing distributed systems is the asynchrony of message-passing communication. Because of asynchrony, sent messages may arrive at their destinations in arbitrary order, leading to confusion and bugs.
Protocols for \emph{causal delivery} of messages~\cite{birman-virtual-synchrony, schiper-causal-ordering, birman-reliable, birman-lightweight-cbcast,raynal1991causal,mattern2005non} bring (partial) order to this chaos by ensuring that when a message $m$ is delivered at a process $p$, any message $m'$ sent before $m$ (in the sense of Lamport's \emph{happens-before} relation~\cite{lamport-clocks}) will have already been delivered at $p$. Applications that rely on causal message delivery include causally consistent data stores~\cite{ahamad-causal-memory, lloyd-cops}, operation-based CRDTs~\cite{shapiro-crdts},
and certain distributed snapshot protocols~\cite{acharya-causal-snapshots, alagar-causal-snapshots}, among others.

\begin{figure*}[htb]

    \centering
    \includegraphics[width=.3\textwidth]{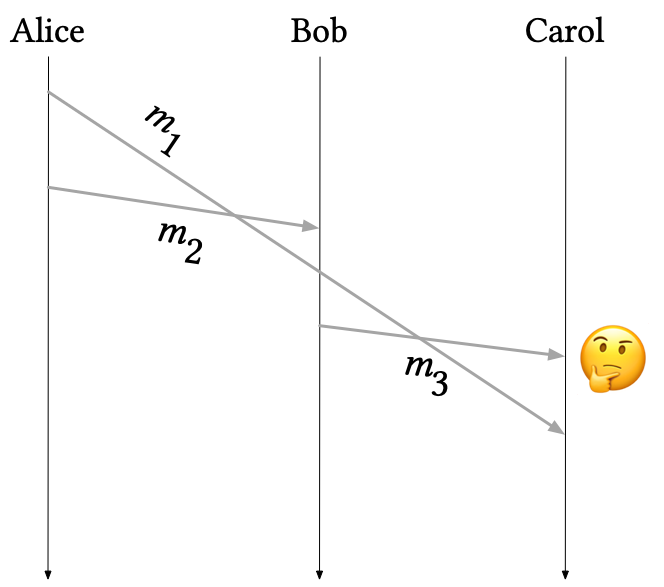}\hfill
    \includegraphics[width=.3\textwidth]{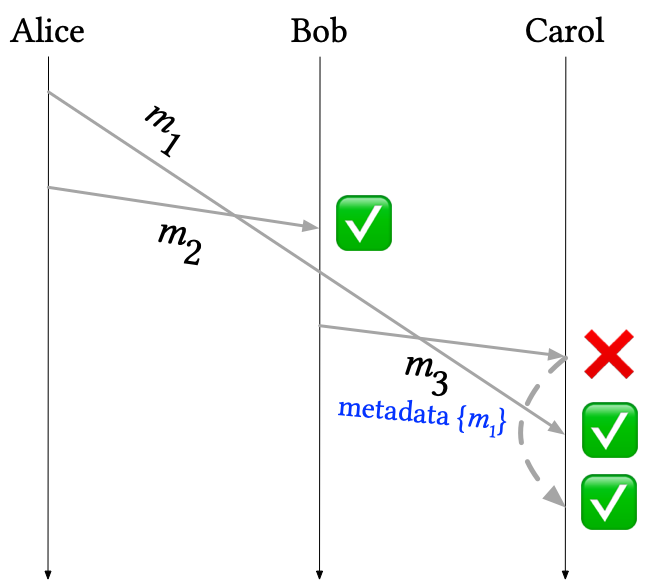}\hfill
    \includegraphics[width=.3\textwidth]{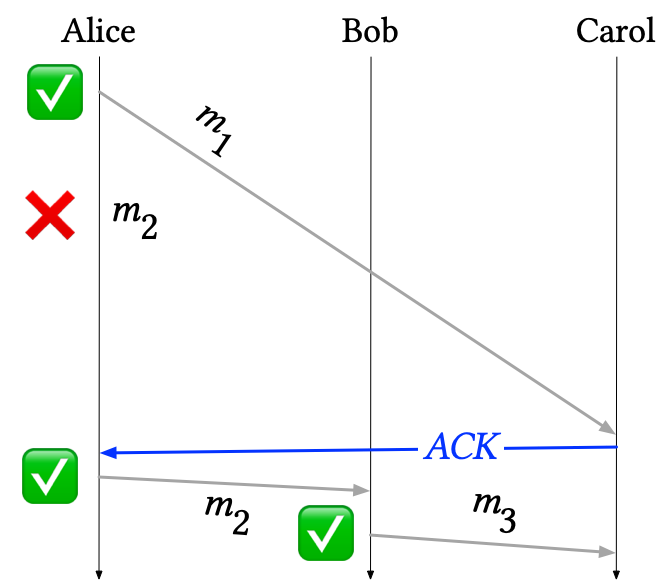}

    \caption{A violation of causal message delivery (left), and the receiver-side (center) and sender-side (right) approaches to preventing it.  The send of message $m_1$ happens before the send of message $m_3$, thanks to the causal relationship created by the intervening message $m_2$.  Causal delivery therefore mandates that Carol must deliver $m_1$ before $m_3$.}
    \label{fig:causal-delivery-intro}

\end{figure*}

The execution in \Cref{fig:causal-delivery-intro} (left) shows an example violation of causal delivery: message $m_1$ is sent by Alice before message $m_3$ is sent by Bob, but Carol receives and delivers $m_3$ before $m_1$.
A standard approach~\cite{birman-lightweight-cbcast,raynal1991causal} to enforcing causal delivery is for senders to attach metadata such as a vector clock~\cite{mattern-vector-time, fidge-vector-time, schmuck-dissertation} or matrix clock~\cite{raynal1991causal} to each message, representing the message's causal dependencies. We call this approach \emph{receiver-side} enforcement of causal delivery.
\Cref{fig:causal-delivery-intro} (center) illustrates (at a high level) the receiver-side approach: when message $m_3$ arrives at Carol, it is deemed undeliverable after inspecting the attached metadata, so it is queued for later delivery.  Eventually, Carol receives and delivers $m_1$, and $m_3$ may be dequeued and delivered.
A disadvantage of the receiver-side approach is that the metadata imposes an overhead on each sent message that is proportional in size to the number of processes in the system. In a system with $n$ participants, depending on the protocol in use, the metadata may be a vector of size $O(n)$ per message~\cite{birman-lightweight-cbcast} or even a matrix of size $O(n^2)$ per message~\cite{raynal1991causal}, which would be prohibitively large for a system with many participants and constrained bandwidth.

Mattern and Fünfrocken~\cite{mattern2005non} propose an alternative approach: instead of potentially delaying delivery on the recipient's side, their protocol delays the \emph{sending} of messages, then sends them without any metadata and unconditionally delivers them on the recipient's side.
\Cref{fig:causal-delivery-intro} (right) illustrates this \emph{sender-side} approach. Alice's message $m_2$ to Bob is deemed unsendable and goes into an outgoing message queue on Alice's end, where it waits until Alice receives an acknowledgement message from Carol. At that point, $m_2$ becomes sendable and is transmitted to Bob. There is no chance of $m_3$ overtaking $m_1$, so it can be immediately delivered when Carol receives it.
Since it does not burden messages with metadata, the sender-side approach is suitable for systems with constrained bandwidth and large numbers of processes (as we quantify later on in \Cref{sec:evaluation}).
The trade-off is that some messages are needlessly delayed. For instance, Alice's message $m_2$ to \emph{Bob} is delayed in the sender-side approach, even though it would \emph{not} violate causal delivery for that message to be sent (and delivered at Bob) earlier.

In this paper, we address the question of how to design a causal delivery protocol that avoids the metadata overhead of the receiver-side approach while also steering clear of needless delays.
We present the \emph{Cykas} (short for ``Can you keep a secret?'') causal message delivery protocol, a sender-side protocol with a new mechanism for \emph{eager} sending of messages.  When circumstances permit, a process $p_i$ may eagerly send a message to a recipient $p_j$ --- but with a caveat that $p_j$ cannot take \emph{externally observable} actions, such as sending messages, until it gets a follow-up message from $p_i$.
Cykas is a drop-in replacement for traditional causal delivery protocols, meaning that it exposes the same interface for sending and delivering messages to clients.  In particular, the mechanism for eager sending of messages is contained entirely within the protocol; clients do not need to make decisions about when a message is to be eagerly sent or not.
The advantage of Cykas over the traditional sender-side protocol is that \emph{intervening} messages, like $m_2$ in \Cref{fig:causal-delivery-intro}, arrive at (and can be immediately delivered at) their recipients sooner. This advantage becomes apparent in executions that involve \emph{long-running jobs} triggered by intervening messages, as we show in \Cref{sec:evaluation}.

In summary, the contributions of this paper are:
\begin{itemize}
    \item The design of the Cykas sender-side causal delivery protocol (\Cref{sec:cykas}).  We motivate the need for the protocol, describe its design, and compare its behavior to the traditional sender-side protocol~\cite{mattern2005non}.
    \item Proofs of safety and liveness for the Cykas protocol (\Cref{sec:correctness}).
    \item An implementation of the Cykas protocol in Rust, using the Stateright~\cite{stateright} actor library (\Cref{sec:verification}). We use the integrated Stateright model checker to carry out bounded verification of safety and liveness for our Cykas implementation.
    \item An empirical performance evaluation of the Cykas protocol (\Cref{sec:evaluation}). We find that Cykas has a performance advantage over traditional sender-side~\cite{mattern2005non} and receiver-side~\cite{raynal1991causal} protocols in large-scale systems with bandwidth constraints, uniform communication patterns, and where messages trigger long-running jobs.
\end{itemize}
\noindent We discuss related work in \Cref{sec:related} and conclude in \Cref{sec:discussion}.

\section{The Cykas Protocol}\label{sec:cykas}

\newcommand{\push}{\textit{push}}
\newcommand{\pop}{\textit{pop()}}
\newcommand{\isempty}{\textit{is_empty()}}
\newcommand{\peek}{\textit{peek()}}
\newcommand{\sender}{\textit{sender}}
\newcommand{\receiver}{\textit{receiver}}
\newcommand{\allzero}{\textit{all\_zero()}}
\newcommand{\iseagersend}{\textit{is\_eager\_send()}}
\newcommand{\esreceiver}{\textit{es\_receiver}}
\newcommand{\bitvecqueue}{\textit{bitvec\_queue}}
\newcommand{\bitvec}{\textit{bitvec}\xspace}

\yt{TODO: In any case, the system model (whether processes can fail or not) should be make explicitly clear early in the text. Still related with the system model, the algorithm seems to assume that the computations performed by processes upon receiving eager messages have no side effects.}

\yt{TODO: starve problem.}

\begin{figure*}[tb]

\centering
\includegraphics[width=\textwidth]{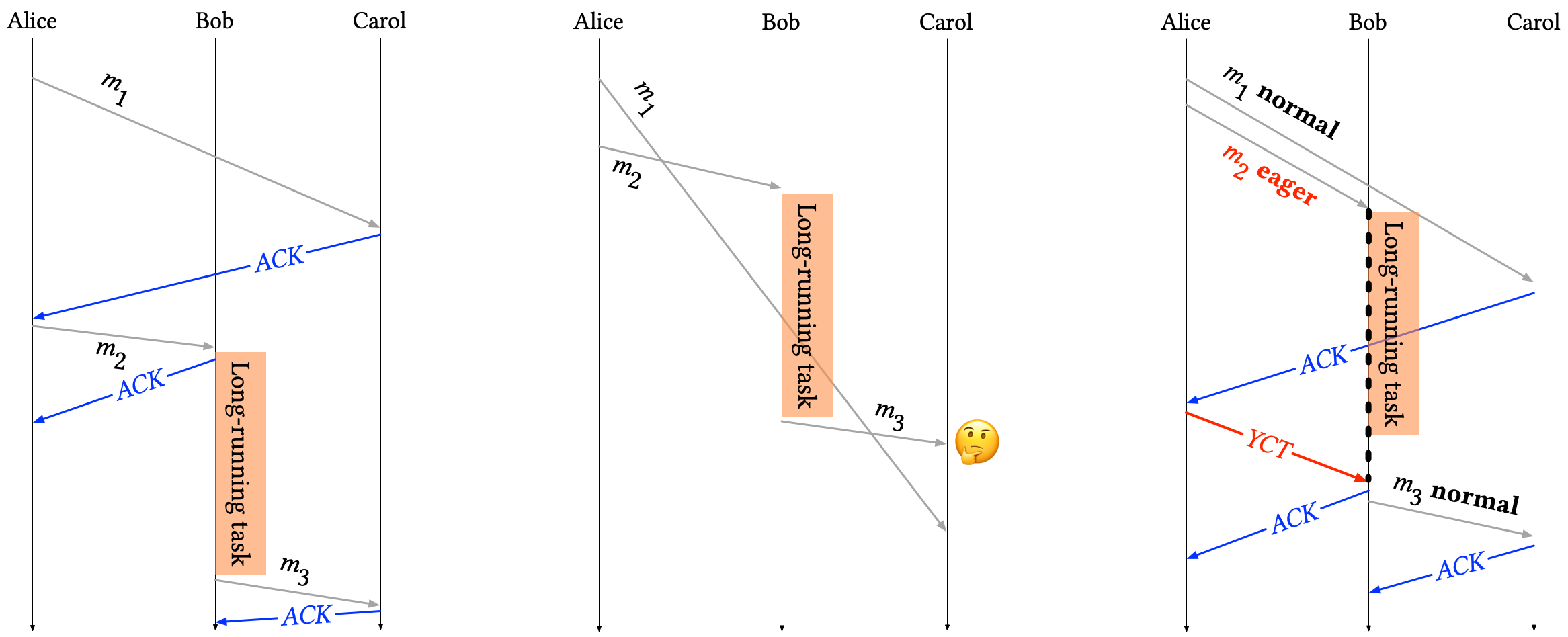}

\caption{In our example execution, the delay imposed by the MFSS protocol is especially problematic (left) if Bob has a long-running job that he cannot begin until receiving $m_2$. However, eagerly sending $m_2$ in a naive way could lead to a violation of causal delivery (center). The Cykas protocol introduces the \EagerSend and \YCT message types, which can be used to let Bob get a head start on the long-running job (right).}
\label{fig:causal-delivery-sender-side-long-running}

\end{figure*}

This section introduces the Cykas protocol.  We motivate Cykas in \Cref{subsec:buffer_issue} by considering a problem with the traditional sender-side approach and explaining at a high level how the Cykas protocol mitigates the problem. \Cref{subsec:cykas_protocol_behavior} describes the Cykas protocol in detail, and \Cref{subsec:sending-in-secret-mode} discusses the reasoning behind one of our protocol design decisions.

\subsection{Motivation: Bob's Long-Running Job}\label{subsec:buffer_issue}

The key idea of the Mattern and Fünfrocken protocol~\cite{mattern2005non} (which we will henceforth call the \emph{MFSS} protocol, for ``Mattern-Fünfrocken sender-side''\footnote{Mattern and Fünfrocken~\cite{mattern2005non} called their protocol the ``buffer protocol'', but we call it MFSS to avoid confusion with other protocols that also involve buffers.}) is that after sending a message, a process must wait for an acknowledgment from the recipient before sending another message.  There are therefore two message types in the MFSS protocol: \emph{normal} messages and \emph{acknowledgement} messages.  On each process, outgoing normal messages go into a FIFO queue, the \emph{output buffer}.  On the recipient's end, messages are received and delivered in receipt order.  A process may place a normal message into the output buffer and then continue running.  Before dequeuing and sending a message in its output buffer, a process waits to receive an acknowledgment message for the sender's last sent message.  For instance, in \Cref{fig:causal-delivery-intro} (right), $m_2$ remains in Alice's output buffer until Alice receives an acknowledgment from Carol for $m_1$.  Acknowledgment messages do not go through the output buffer; they are sent immediately in response to a normal message being received and delivered.

Compared with a receiver-side approach, the MFSS protocol imposes unnecessary delays on delivery of intervening messages. For instance, in \Cref{fig:causal-delivery-intro} (right), the delivery of $m_2$ is delayed at Bob (due to $m_2$ being sent later).
This delay is particularly problematic if Bob has some \emph{long-running job} to do as a result of getting $m_2$, as shown in \Cref{fig:causal-delivery-sender-side-long-running} (left).  Suppose that Bob cannot get started on the job until receiving $m_2$, and he must complete it before sending $m_3$.
Unfortunately, in the MFSS protocol, Alice cannot send $m_2$ until she receives the \ACK for $m_1$, and so Bob would have to wait to start on the long-running job.
But if Alice sent $m_2$ before getting the \ACK for $m_1$, hoping to give Bob a head start, we would risk violating causal delivery since Bob's $m_3$ could still arrive at Carol before $m_1$, as shown in \Cref{fig:causal-delivery-sender-side-long-running} (center).
Hence, Bob must be prevented from sending $m_3$ too early, but he should be allowed to start the long-running job early, because it only affects his internal state. What we want, then, is a way to let Bob receive $m_2$ and get started early on the long-running job, but without the risk of Bob sending $m_3$ too early.

The Cykas protocol lets us sidestep the conservative behavior of the MFSS protocol by differentiating two ways of sending messages: \NormalSend and \EagerSend. A \NormalSend works the same as in the MFSS protocol; it contains a payload and will be delivered as soon as it is received. An \EagerSend also contains a payload and will be delivered upon receipt. However, upon delivering an \EagerSend message sent from a process $p_i$, the recipient must ``keep secret'' the information that it received in the message \emph{until} it receives a follow-up message from $p_i$ that we call a \YCT message, which is short for ``you can tell.''

In all, there are four types of messages in the Cykas protocol: \NormalSend, \EagerSend, \YCT, and \ACK. Receipt of a \NormalSend imposes no restrictions on the receiving process. Upon receiving an \EagerSend message, however, a process goes into \emph{secret mode}.  In secret mode, a process may take internal actions and receive messages, but \emph{cannot send any messages} except \ACK and \YCT. A process that sends an \EagerSend must follow it with a corresponding \YCT message to the same recipient, which is used to unblock the recipient from sending messages. A process leaves secret mode, returning to \emph{normal mode}, when it has received \YCT messages for any \EagerSend messages it previously received. Finally, a process that receives any of the above message types sends \ACK back to the sender to acknowledge the receipt and delivery of the message.

\Cref{fig:causal-delivery-sender-side-long-running} (right) shows an example run of the Cykas protocol.
First, Alice sends $m_1$ as a \NormalSend message to Carol. 
Instead of waiting for the \ACK of that message, Alice sends $m_2$ as an \EagerSend to Bob, triggering his long-running job.
Bob receives and delivers $m_2$ and transitions to secret mode, indicated by the dashed line on Bob's process line.
After finishing his long-running job, Bob is still (briefly) prevented from sending $m_3$ to Carol, since he is running in secret mode and hasn't received \YCT yet.
After receiving \ACK{}s of both the \NormalSend message to Carol and the \EagerSend message to Bob, Alice sends a \YCT message to Bob. 
At this point, Bob transitions back to normal mode 
and can send $m_3$ to Carol without risk of violating causality. 

\subsection{Cykas in Detail}\label{subsec:cykas_protocol_behavior}

Having established intuition for the Cykas protocol's behavior, we now describe the protocol precisely. We assume a fixed set of $N$ processes $p_1, ..., p_N$ running the protocol, which we call \emph{Cykas processes}. We further assume that message delivery is reliable.

\subsubsection{Local data structures} Each Cykas process maintains the following local data structures:
\begin{itemize}
    \item \OB: the \emph{output buffer}, a FIFO queue of unsent messages, initially empty. Whenever a Cykas process receives an \ACK or \YCT message, and whenever a new message is inserted into the \OB, the Cykas process will determine whether the message at the \emph{head} of the $ OB $ can be sent, and if so, whether it can be a \NormalSend or must be an \EagerSend.
    \item \UNACK: a bitvector of length $N$. If a Cykas process $p_i$ has $\UNACK[j] = 1$, it means $p_i$ has sent a message to process $p_j$ and hasn't yet received an \ACK for that message. Initially, all entries in \UNACK are 0 because no messages have been sent.
    \item \MODE: a nonnegative integer, initially $0$, representing the number of \YCT messages currently being awaited. If $\MODE=0$, the process is in \emph{normal mode}; otherwise it is in \emph{secret mode}.
    \item \EAGERSENT: a map, initially empty, from process identifiers to \emph{queues of bitvectors}, used to track when to send a \YCT message for each \EagerSend. Each process keeps an entry in its \EAGERSENT{} map for each process $p_j$ to which it has sent at least one \EagerSend message. Process identifier $p_j$ maps to a queue of length-$N$ bitvectors, one for each \EagerSend that has been sent to $p_j$. In each bitvector, the bit at index $j'$ represents whether the process is waiting for an \ACK from $p_{j'}$ to send a \YCT. For example, in the execution in \Cref{fig:causal-delivery-sender-side-long-running} (right), Alice's \EAGERSENT{} map immediately after sending $m_2$ to Bob is $\{ \textit{Bob} \rightarrow \langle [ 0 0 1 ] \rangle \}$, indicating that Alice is waiting to get an \ACK from Carol before she can send a \YCT to Bob.
\end{itemize}

\subsubsection{Protocol operation}

Cykas exposes two operations to clients, \emph{application-send} and \emph{deliver}, where application-send is triggered by the client and deliver is a callback operation triggered by the underlying protocol.  Consider a single message $m$ passing between $p_{i}$ and $p_{j}$. The following operations occur:
\begin{itemize}
    \item \emph{application-send} (\Cref{algo:cykas_application_send}): process $p_{i}$ initiates sending $m$ by inserting $m$ into its \OB.
    \item \emph{network-send} (\Cref{algo:cykas_process_send_m}): message $m$ leaves the \OB, wrapped in either \NormalSend or \EagerSend, and is transmitted to $p_{j}$. We assume that the underlying network transport mechanism provides a {\sc NetworkSend} function that takes a message and a destination. The Cykas protocol controls if and when a network-send can happen, and what type of wrapper to be used.
    \item \emph{receive-deliver} (\Cref{algo:cykas_process_receive_m}): $m$, wrapped in \NormalSend or \EagerSend, arrives at $p_{j}$. $p_{j}$ unwraps $m$, immediately delivers its payload to the application, and network-sends an \ACK message back to $p_{i}$.
    \item \emph{receive-\ACK} (\Cref{algo:cykas_process_receive_ack}): an \ACK message sent from $p_{j}$ is received by $p_{i}$. The protocol will then attempt to send any application messages and \YCT{}s that had been blocked on receipt of the \ACK.  \ACK messages are not delivered to the application, only used by the Cykas protocol.
    \item \emph{receive-\YCT} (\Cref{algo:cykas_process_receive_yct}): if $m$ was an \EagerSend, $p_i$ will eventually follow up with a \YCT message to $p_j$. This operation runs when a \YCT message is received. \YCT messages are not delivered to the application, only used by the Cykas protocol.
\end{itemize}

\begin{algorithm}
  \caption{application-send message $m$}\label{algo:cykas_application_send}
  \begin{algorithmic}[1]

    \State \OB.\push($m$)

    \State \Call{try_send_message}{}

  \end{algorithmic}
\end{algorithm}

\begin{algorithm}
    \caption{attempt to network-send message $m$}\label{algo:cykas_process_send_m}
    \begin{algorithmic}[1]
    \Function{try_send_message}{}
        \If{$\MODE > 0$}
            \Return \Comment{no send in secret mode}
        \EndIf

        \While{$\lnot \OB.\isempty$}
            \State $m \gets \OB.\peek$
            \State $i \gets m.\sender, j \gets m.\receiver$
            \If{$\UNACK[j] = 1$}
                \Return \Comment{cannot send}
            \EndIf

            \If{\UNACK.\allzero}
            
                \State \Call{NetworkSend}{\textit{NormalSend}($m$), $m$.\receiver}
            \Else
                \State \Call{NetworkSend}{\textit{EagerSend}($m$), $m$.\receiver}
                \State $\EAGERSENT[j].\push(\UNACK)$
            \EndIf

            \State $\OB.\pop$
            \State $\UNACK[j] \gets 1$
            
        \EndWhile
    \EndFunction

    \end{algorithmic}
\end{algorithm}

\begin{algorithm}
    \caption{receive-deliver message $m$}\label{algo:cykas_process_receive_m}
    \begin{algorithmic}[1]
  
        \If{$m$.\iseagersend}

            $\MODE \gets \MODE + 1$  \Comment{now in secret mode}
        \EndIf

        \State \Call{Deliver}{$m$}
        \State \Call{NetworkSend}{\textit{ACK}($m$), $m$.\sender}
    \end{algorithmic}
\end{algorithm}

\begin{algorithm}
    \caption{receive-\ACK}\label{algo:cykas_process_receive_ack}
    \begin{algorithmic}[1]

        \State $i \gets \ACK.\sender$
        \State $\UNACK[i] \gets 0$
        \For{\textbf{each} $(\esreceiver, \bitvecqueue)$ \textbf{in} \EAGERSENT}
            \For{\textbf{each} \bitvec \textbf{in} \bitvecqueue}
                \State $\bitvec[i] \gets 0$
            \EndFor
            \State \Call{try_send_yct}{\esreceiver}
            \State \Call{try_send_message}{}
        \EndFor \\

        \Function{try_send_yct}{\receiver}
            \While{$\lnot \EAGERSENT[\receiver].\isempty$}
                \State $\bitvec \gets \EAGERSENT[\receiver].\peek$
                \If{$\lnot \bitvec.\allzero \lor \UNACK[\receiver] = 1$}
                    \State \textbf{break}
                \EndIf
                
                \State \Call{NetworkSend}{\YCT, \receiver}
                \State \EAGERSENT[\receiver].\pop
            \EndWhile
        \EndFunction
    \end{algorithmic}
\end{algorithm}

\begin{algorithm}
    \caption{receive-\YCT}\label{algo:cykas_process_receive_yct}
    \begin{algorithmic}[1]

        \State $MODE \gets MODE - 1$
        \State \Call{try_send_message}{}
  
    \end{algorithmic}
\end{algorithm}

The behavior shown in \Cref{algo:cykas_process_send_m} is crucial to the protocol on the sender's side. A Cykas process runs in one of two modes: \emph{normal mode} and \emph{secret mode}.  Cykas processes begin in normal mode. If a process $p$ running in normal mode network-sends a message to a recipient but has not received an \ACK, $p$ cannot network-send a message to the same recipient until receiving the \ACK. However, it may network-send a message to another recipient. If $p$ has network-sent any un-\ACK{}ed  messages, the new message will be an \EagerSend (e.g., $m_2$ in \Cref{fig:causal-delivery-sender-side-long-running}); otherwise, it will be a \NormalSend (e.g., $m_1$ and $m_3$ in \Cref{fig:causal-delivery-sender-side-long-running}).

When a process receives an \EagerSend from another process (\Cref{algo:cykas_process_receive_m}, line 1), it transitions to secret mode (if it is not already in secret mode), in which neither \NormalSend nor \EagerSend are allowed; it can only network-send \ACK and \YCT messages. Once all \EagerSend{}s that the process has received are unblocked by corresponding \YCT{}s, the process transitions back to normal mode, in which it can freely send messages again.

On the recipient's side, the behavior shown in \Cref{algo:cykas_process_receive_ack} is crucial to the protocol.  When a process receives an \ACK, it network-sends any \YCT messages it is able to send, and attempts to send any new application messages queued in its \OB.

Although both \NormalSend and \EagerSend messages are part of the Cykas protocol internally, only the \emph{application-send} interface is exposed to clients of the protocol. The choice between \NormalSend and \EagerSend is made by the protocol itself: a process running in normal mode will send messages normally when possible, and eagerly otherwise. The generic \emph{application-send} interface makes it possible to use the Cykas protocol as a drop-in replacement for existing causal delivery protocols.

\subsection{Why disallow sending in secret mode?}
\label{subsec:sending-in-secret-mode}

\begin{figure}[b]
    \centering
    \includegraphics[width=0.35\textwidth]{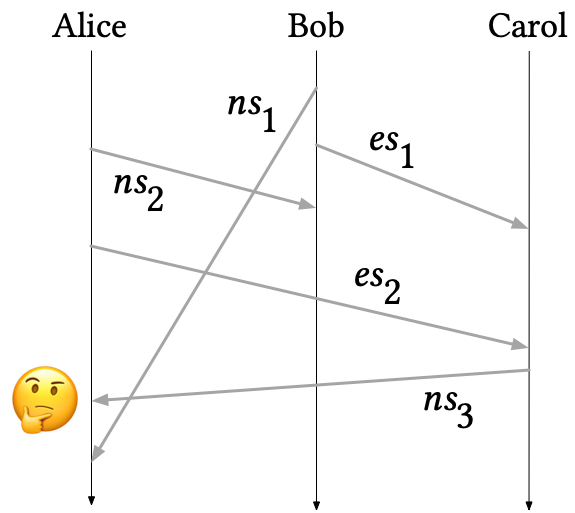}
    \caption{If sending is allowed in secret mode, a causality violation can occur.  \ACK and \YCT messages are elided in this figure.}
    \label{fig:origin_cykas_violation}
\end{figure}

A Cykas process is considerably constrained when it is running in secret mode (that is, when it has received at least one \EagerSend message and not yet received the corresponding \YCT message). One might wonder if a process in secret mode should be able to at least correspond with the sender of its \emph{most recently delivered} \EagerSend message.  (After all, if someone tells you a secret, you might expect to be able to safely tell a secret back to them.)
However, allowing communication with an eager sender would admit causality violations.  \Cref{fig:origin_cykas_violation} illustrates the problem with this more permissive approach. Here, messages with names $\mathit{ns}_i$ are \NormalSend{}s, while messages with names $\mathit{es}_i$ are \EagerSend{}s. For simplicity, we elide \ACK and \YCT messages.

Bob sends $\mathit{ns}_1$ to Alice, and then sends $\mathit{es}_1$ to Carol.  Carol delivers $\mathit{es}_1$ and enters secret mode, during which she can only send messages to Bob. Meanwhile, Alice sends $\mathit{ns}_2$ to Bob and then $\mathit{es}_2$ to Carol. Carol delivers $\mathit{es_2}$ and can now only send messages to Alice. While still in secret mode, Carol sends $\mathit{ns}_3$ to Alice, who delivers it.
This violates causal delivery, since the send of $\mathit{ns}_1$ happened before the send of $\mathit{ns}_3$ (send($\mathit{ns}_1$) $\rightarrow$ send($\mathit{es}_1$) $\rightarrow$ deliver($\mathit{es}_1$) $\rightarrow$ send($\mathit{es}_3$)), but Alice delivers $\mathit{ns}_3$ before $\mathit{ns}_1$. (Alice's message $\mathit{ns}_2$ is not directly involved in the causality violation, but is present in the execution because it is the reason why $\mathit{es}_2$ must be sent eagerly.) We found this counterexample using the Stateright model checker, which we discuss further in \Cref{sec:verification}.

\section{Proof of Correctness}\label{sec:correctness}

In this section, we consider the correctness of the Cykas protocol in terms of its safety (\Cref{subsec:safety_proof}) and liveness (\Cref{subsec:liveness_proof}) properties.  The safety property we wish to ensure is causal delivery, i.e., messages are never delivered in an order that violates the causal order.  The liveness property we wish to ensure is that, assuming a reliable network, all messages will eventually be delivered (rather than, for example, being delayed indefinitely in a sender's output buffer).
We borrow some notation previously established by Mattern and Fünfrocken~\cite{mattern2005non} in our proof, with a few differences specific to our setting:
\newcommand{\ns}{\mathit{ns}}
\newcommand{\rd}{\mathit{rd}}
\begin{itemize}
  \item $a \sim b$: events $a$ and $b$ happen at the same process.
  \item $a \rightarrow b$: event $a$ happened before event $b$, in the sense of Lamport's happens-before relation~\cite{lamport-clocks}. Note that $\to$ is a transitive relation.
  \item $OB_{i}$: output buffer of process $p_{i}$.
  \item $s_{i}^{m}$: application-send of message $m$ at $p_{i}$.
  \item $\ns_{i}^{m}$: network-send of message $m$ at $OB_{i}$.
  \item $\rd_{i}^{m}$: receive-deliver of message $m$ at $p_{i}$.
  \item $r^{\ACK(m)}_{i}$: receive \ACK of message $m$ at $p_{i}$.
  \item $r^{\YCT}_{i}$: receive \YCT at $p_{i}$.    
  \item $M_{i \rightarrow j}$: a message sent from $p_{i}$ to $p_{j}$.
\end{itemize}

 If it is clear from the context which process an event takes place on, we will omit the subscript, e.g., we will write just $s^{m}$ instead of $s_{i}^{m}$.

\subsection{Safety}\label{subsec:safety_proof}

We define causally ordered delivery as follows: given an execution with distinct send events $s^{m}$, $s^{m'}$ and their corresponding receive-deliver events $\rd^{m}$, $\rd^{m'}$, the following always holds:
\begin{gather}
  s^{m} \rightarrow s^{m'} \land \rd^{m} \sim \rd^{m'} \Longrightarrow \rd^{m} \rightarrow \rd^{m'}
  \tag{Causally Ordered Delivery}
\end{gather}

\noindent That is, if $m$ is sent before $m'$ and $m$ and $m'$ are received and delivered at the same process, then the delivery of $m$ precedes the delivery of $m'$.

\begin{proposition}\label{prop:send_after_deliver}
    Any execution of Cykas with distinct send events $s^{m}$, $s^{m'}$ and corresponding receive-deliver events $\rd^{m}$, $\rd^{m'}$
    satisfies causally ordered delivery.
\end{proposition} 

The proof idea relies on two key aspects of the Cykas protocol. First, if a process $p_i$ has not received an $\ACK$ for some message, then any subsequent messages sent by $p_i$ are \emph{eager} and thus must block the receiving processes. Second, if $p_i$ sends $m$ to process $p_j$, and $p_i$ has not received an $\ACK(m)$, then $p_i$ cannot send any more messages (eager or otherwise) to $p_j$ until it receives the missing $\ACK(m)$.

\begin{proof}
  Let $M_{i \to k}$ and $M_{j \to k}'$ be two messages in a given execution satisfying the hypothesis $s^{M} \rightarrow s^{M'} \land \rd^{M} \sim \rd^{M'}$. Note the $i = j$ case is immediate from \Cref{algo:cykas_process_send_m}, since if $p_i$ sends $M$ to $p_k$, then $\UNACK_{i}[k] \leftarrow 1$, and by Line 7, process $p_i$ cannot send $M'$  to $p_k$ until $p_i$ receives an $\ACK(M)$, which implies an $\rd_{k}^{M}$ event (\Cref{algo:cykas_process_receive_m}), and therefore causally ordered delivery.

  For the $i \neq j$ case, assume towards contradiction that we have a causal violation. That is, process $p_k$ executed $\rd^{M'}$, then $\rd^{M}$, so that
  $\rd^{M'} \to \rd^{M}$.

  Since $s^{M} \to s^{M'}$, we have the network sends $\ns_{i}^{M} \to \ns_{j}^{M'}$.
  Since network-sending $M$ causally precedes network-sending $M'$, by transitivity of $\to$, there must be a causal chain of events
  on $\ell$ processes where $p_i$ sent some message $z_{0}$ to a process $p_{i_1'}$, who delivered $z_0$, then sent some
  message $z_1$ to a process $p_{i_2'}$, who delivered $z_1$, and so on, until a process $p_{i_\ell'}$ sent some message $z_\ell$ to
  $p_j$, thus transitively carrying $M$ as a potential cause for $M'$. More precisely, we have the causal chain
  \[(\ns_{i}^{M} \to \ns_{i}^{z_0})  \to (\rd_{i_1'}^{z_0} \to \ns_{i_1'}^{z_1}) \to \cdots \to \]
  \[(\rd_{i_\ell'}^{z_{\ell-1}} \to \ns_{i_\ell'}^{z_\ell}) \to (\rd_{j}^{z_\ell} \to \ns_{j}^{M'}) \to \rd_{k}^{M'}.\]
  Since by assumption we have $\rd_{k}^{M'} \to \rd_{k}^{M}$, there is no $\rd_{k}^{M}$ event in the above causal chain,
  which means $p_k$ never executed \Cref{algo:cykas_process_receive_m} for $M$ during the above sequence of events.
  Therefore, process $p_i$ never received an $\ACK(M)$ from $p_k$, and by \Cref{algo:cykas_process_send_m}, we have $\UNACK_{i}[k] = 1$
  during the above sequence of events.

  From \Cref{algo:cykas_process_send_m}, it must be the case that $p_i$ \emph{eagerly} sent $z_0$ to $p_{i_1'}$,
  setting $\MODE > 0$ and placing $p_{i_1'}$ in \emph{secret mode}, by \Cref{algo:cykas_process_receive_m}.
  Observe that line 2 of \Cref{algo:cykas_process_send_m} disallows any further network-sends until the process receives a $\YCT$ message.
  Therefore, $p_{i_1'}$ must have executed a $\rd_{i_1'}^{\YCT}$ event before the $\ns_{i_1'}^{z_1}$ event.

  But this is a contradiction, since by \Cref{algo:cykas_process_receive_ack}, process $p_i$ cannot send a $\YCT$ until $\UNACK_{i}$ is all zeroes,
  which cannot happen since $p_k$ does not execute $\rd^{\ACK(M)}$ in the above chain of events, thus establishing that a causal violation
  was not possible.
\end{proof}

\subsection{Liveness}\label{subsec:liveness_proof}
As we said at the beginning of this section, liveness of Cykas means every message will eventually be delivered. To prove liveness, we are obligated to assume that the network is reliable (i.e., messages that are sent will be delivered eventually). Given that assumption, we want to prove the following proposition:

\begin{proposition}\label{prop:cykas_eventually_send}
    All messages in an \OB will be network-sent eventually.
\end{proposition}

From \Cref{algo:cykas_process_send_m}, we see that for a message $M_{i \rightarrow j}$, its network-send event $\ns^{M}$ requires that $\UNACK[j] = 0$ and $\MODE = 0$. We will prove that in a reliable network, this will eventually be the case.

\begin{lemma}\label{lemma:unack_liveness}
    Every bit in \UNACK will be set to $0$ eventually.
\end{lemma}

\begin{proof}
    According to \Cref{algo:cykas_process_send_m} and \Cref{algo:cykas_process_receive_ack}, bits in \UNACK are set to $1$ by sending messages and will be set back to $0$ when receiving their corresponding \ACK{}s. From \Cref{algo:cykas_process_receive_m}, we know that an \ACK{} is network-sent as soon as a message is received. As we assume reliable delivery, \ACK{}s will be delivered eventually, which means every bit in \UNACK will be set to $0$ eventually.
\end{proof}

\begin{lemma}\label{lemma:eager_sent_liveness}
  Every \EagerSend{} will eventually be followed by its corresponding \YCT{}.
\end{lemma}

\begin{proof}
    As we have proved in \Cref{lemma:unack_liveness}, \ACK{}s for each message will eventually be sent and delivered. This means that, according to \Cref{algo:cykas_process_receive_ack}, every bit in the bitvectors stored in \EAGERSENT{} will be set to $0$ eventually, and thus all their corresponding \YCT{}s will be sent eventually.
\end{proof}

\begin{lemma}\label{lemma:mode_liveness}
    \MODE will become $0$ eventually.
\end{lemma}

\begin{proof}
    According to \Cref{algo:cykas_process_receive_m} and \Cref{algo:cykas_process_receive_yct}, \MODE is incremented when receiving an \EagerSend and decremented when receiving a \YCT{}. In \Cref{lemma:eager_sent_liveness} we have proved that \YCT{}s for every \EagerSend{}s will be sent eventually. Thus, $MODE$ will be set to $0$ eventually.
\end{proof}

From \Cref{lemma:unack_liveness} and \Cref{lemma:mode_liveness} we have that, in a reliable network, at any given process, eventually $\UNACK[j] = 0$ for all $j$ and $\MODE = 0$. Therefore, said process will eventually be able to network-send any messages that have been inserted into its \OB, and \Cref{prop:cykas_eventually_send} holds.

\section{Bounded Verification}\label{sec:verification}

In this section, we describe bounded verification of a Rust implementation of the Cykas protocol. We used Stateright~\cite{stateright}, an implementation-level model-checking library for Rust, to check the safety and liveness properties of the protocol for a bounded class of executions.

\begin{table*}[htbp]
\caption{Model checking results for Cykas and two other protocols, with 3 actors and 3 application messages per actor}
\centering
\begin{tabular}{lrrrrr}
\toprule
Protocol & Total/unique states & Max depth & Time (s) & Correct? \\
\midrule
MFSS protocol~\cite{mattern2005non} & 1,543M / 527M & 37 & 163 & \ding{51} \\
Cykas w/ secret-mode sends~(\S\ref{subsec:sending-in-secret-mode}) & 44M / 13M & 43 & 5 & \ding{55} \\
Cykas protocol & 5,974M / 1,945M & 49 & 709 & \ding{51} \\
\bottomrule \\
\end{tabular}
\label{tab:protocol_check_results}
\end{table*}

In Stateright, an \emph{actor system} consists of multiple actors connected by a network, and property checkers that check if specified properties are violated during model checking.  In our case, actors are Cykas processes. We implemented custom causality and liveness checkers in Stateright for our Cykas implementation.

Our causality checker leverages the fact that vector clocks~\cite{mattern-vector-time, fidge-vector-time, schmuck-dissertation} precisely capture Lamport~\cite{lamport-clocks}'s happens-before relation: for any messages $m_1, m_2$ in an execution, $send(m_1) \rightarrow send(m_2)$ if and only if the vector clock of $m_1$ is less than the vector clock of $m_2$~\cite{mattern-vector-time}. Our causality checker associates vector clocks with sent messages. So that we do not defeat the purpose of sender-side enforcement of causal delivery, rather than piggybacking vector clocks directly on messages in the protocol (as a receiver-side protocol would do), the vector clocks are maintained in a separate data structure by the causality checker.

The causality checker maintains a history of each execution of the actor system, which includes the vector clock of each actor and a delivery history where each actor stores its delivered messages along with their vector clocks.
When an actor $p$ sends a message, the causality checker updates $p$'s vector clock and associates the vector clock with the message.
When an actor $p$ delivers a message $m$, the causality checker appends $m$ and its vector clock to $p$'s delivery history, and updates $p$'s vector clock.
At the end of each action, the causality checker inspects the delivery history for each actor. If the vector clocks of the delivered messages are out of order, the checker reports a violation of causal delivery.

When checking for causality violations, our causality checker takes only application messages -- that is, \NormalSend and \EagerSend messages -- into account.  \YCT and \ACK messages can be ignored by the causality checker because they are part of the causal delivery protocol itself and not visible to applications.

Our liveness checker maintains a counter for message deliveries in an execution. Each time a message is delivered, the counter is incremented. During model checking, if the current state is a final state, the liveness checker checks whether the number of delivered messages matches the number of application-sent messages. (In particular, we want to avoid a situation in which an application-sent message gets stuck in the \OB and is never network-sent.) If not, the checker reports a liveness violation.

We use an AWS EC2 \emph{c7i.48xlarge} instance with 96 threads to execute model-checking jobs in parallel. \Cref{tab:protocol_check_results} (row 3) shows the results of model-checking the safety and liveness of the Cykas protocol on a system of three actors, with three messages sent per actor.  Since our bounded liveness checker only checks the final state of each execution, the running times we report in \Cref{tab:protocol_check_results} are almost entirely due to the safety checker.  We found no violations of safety or liveness for the Cykas protocol.  We also checked the correctness of an implementation of the MFSS~\cite{mattern2005non} protocol up to the given bounds (\Cref{tab:protocol_check_results}, row 1). Finally, we checked a variant of the protocol that allows communication with an eager sender while in secret mode (\Cref{tab:protocol_check_results}, row 2), and found that this variant \emph{does} violate causal delivery, as discussed in \Cref{subsec:sending-in-secret-mode}.  The check of the buggy protocol is relatively fast because Stateright stops exploring states as soon as it finds a property violation.

\section{Evaluation}\label{sec:evaluation}

In this section, we evaluate the performance of the Cykas protocol compared to two other causal delivery protocols: the MFSS protocol~\cite{mattern2005non}, described in \Cref{subsec:buffer_issue}, and a classic receiver-side protocol due to Raynal et al.~\cite{raynal1991causal}, henceforth the \emph{matrix} protocol.
Our evaluation seeks to answer: under what system characteristics, long-running job characteristics, and communication patterns does Cykas outperform the other two protocols?

We describe our experimental setup in \Cref{subsec:exp_setup}, then identify where Cykas excels through three analyses. First, we quantify the scalability limitations of the receiver-side approach compared to Cykas and MFSS~(\Cref{subsec:matrix_scalability}). Second, we examine Cykas's performance advantage over MFSS under varying job characteristics~(\Cref{subsec:perf_analysis}). Third, we analyze how different communication patterns affect Cykas's performance~(\Cref{subsec:random_vs_hotspot}).

\subsection{Experimental Setup}\label{subsec:exp_setup}

We run our experiments on a simulated network to isolate the behavior of the protocols and avoid potential interference from network variability.
We represent message-passing programs as partially-ordered sets of message send events, then simulate execution of these programs using each of the protocols.
Messages can trigger long-running computational jobs upon delivery.
We vary the following parameters across experiments: number of processes, network characteristics (network bandwidth and message propagation delay), job length, communication frequency, and hotspot percentages (see \Cref{subsec:random_vs_hotspot}) in the system.
For each execution, we measure \emph{total execution time} (i.e., time to complete all protocol communications), and \emph{average job start time} (the average time at which long-running jobs in the execution begin after their triggering messages are delivered).
The latter metric provides insight into whether Cykas enables earlier job initiation through its eager sending mechanism.

\subsection{Scalability of Sender-side vs. Receiver-side Protocols}\label{subsec:matrix_scalability}

We first establish the conditions under which sender-side protocols become preferable to receiver-side protocols.
Receiver-side protocols like the matrix protocol~\cite{raynal1991causal} embed causal metadata directly in application messages, requiring metadata of size $O(n^2)$ per message for a system of $n$ processes in the worst case.\footnote{In the special case where all messages are broadcast to all participants, we can use vector clocks of size $O(n)$~\cite{birman-lightweight-cbcast}.  The three protocols we consider here all support arbitrary point-to-point messaging patterns, rather than only broadcast messages.}
Compared to receiver-side protocols, sender-side protocols exchange additional control messages (\ACK, \YCT), but impose no metadata overhead per message.
Under bandwidth constraints, metadata size creates a bottleneck as the number of processes in the system increases.
To quantify this bottleneck, we evaluate how execution time varies with the process-to-bandwidth ratio across different system scales.

\begin{figure*}[tb]
    \centering
    \includegraphics[width=0.8\textwidth]{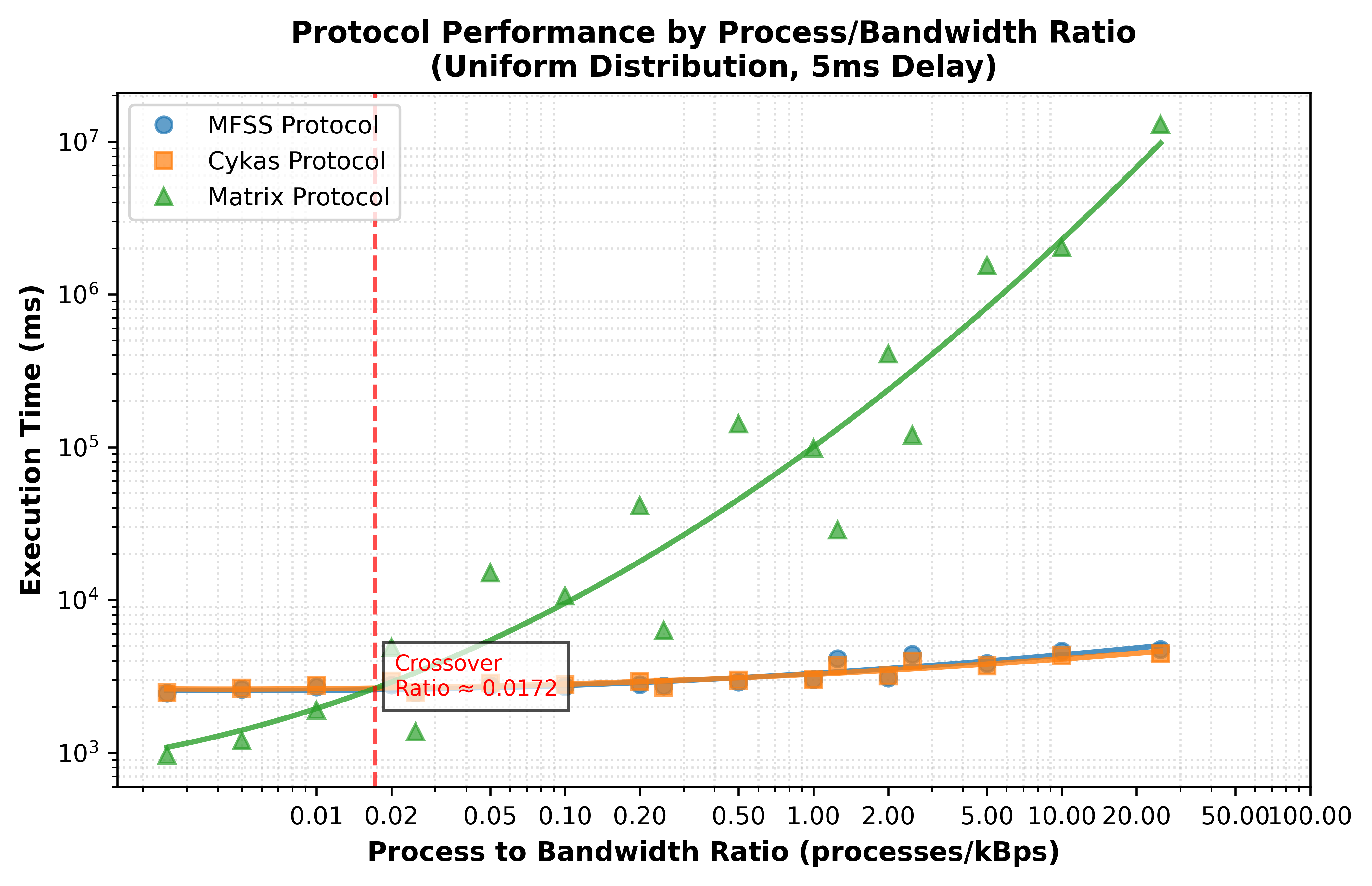}

    \caption{Execution time comparison for the MFSS, Cykas, and matrix protocols at varying process-to-bandwidth ratios. With increasing process-to-bandwidth ratio, the matrix protocol's execution time increases quadratically.}
    \label{fig:matrix_scalability}
\end{figure*}

\Cref{fig:matrix_scalability} illustrates how execution times of the matrix, Cykas, and MFSS protocols vary with the process-to-bandwidth ratio, expressed as processes/kBps. For these experiments, we set the number of processes to 25, 50, 100, 200, and 500, and the bandwidth to 20, 100, 1000, and 10000 kBps. We fixed the propagation delay to 5 ms and configured each process to send 100 messages.
With increasing process-to-bandwidth ratio, the matrix protocol's execution time increases quadratically, with a crossover point around a ratio of 0.02 processes/kBps where the matrix protocol becomes slower than the two sender-side protocols, due to the matrix protocol's metadata overhead. For example, in a system with 5000 kBps available bandwidth and a propagation delay of 5 ms, the matrix protocol would be the better choice only if the number of processes in the system is smaller than 100 (i.e., 5000 kBps $\times$ 0.02 processes/kBps).
The absolute value of the crossover point should be treated as illustrative rather than definitive; optimizing the implementation of receiver-side protocols in specific scenarios would shift the crossover point. Nevertheless, we see that the matrix protocol scales poorly under lower bandwidth conditions and large numbers of participants, due to its metadata-heavy messages, and sender-side protocols provide more scalable performance as system size grows or bandwidth decreases.

This experiment shows the system characteristics where sender-side protocols excel: systems with bandwidth constraints, large numbers of processes, and low propagation delay.
For the rest of this section, we focus on systems with these characteristics and now turn to a comparison of the two sender-side protocols, Cykas and MFSS.

\subsection{Cykas vs MFSS: Long-Running Job Performance Analysis}\label{subsec:perf_analysis}

\begin{figure*}[htb]
    \centering
    \includegraphics[width=\textwidth]{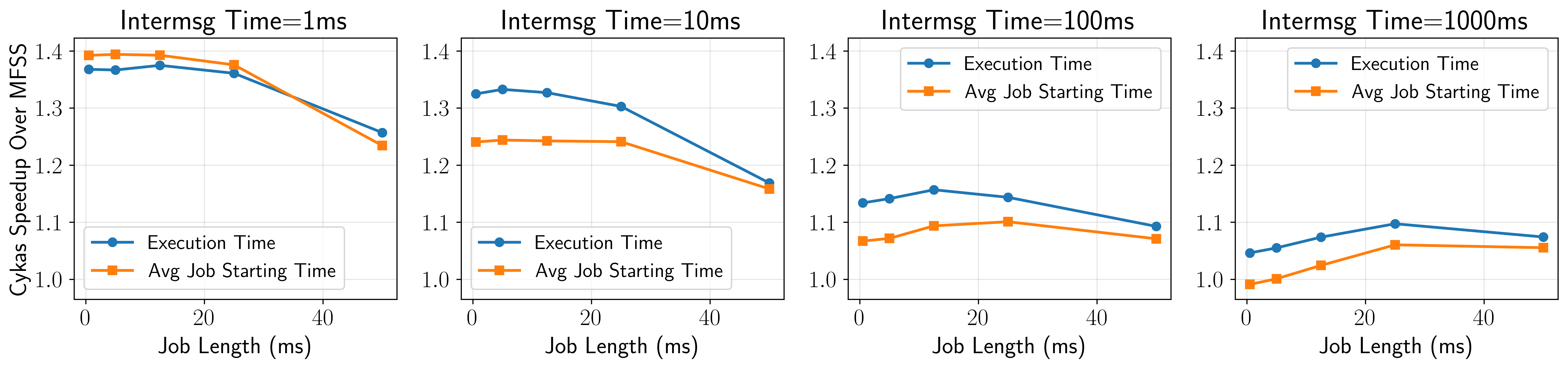}

    \caption{Cykas speedup over MFSS across varying job characteristics (higher is better). Each subfigure shows a different communication frequency, with job length on the x-axis and Cykas speedup over MFSS on the y-axis (values > 1 indicate Cykas outperforms MFSS). Blue lines show execution time speedup, orange lines show average job start time speedup.}
    \label{fig:job_impact}
\end{figure*}

\begin{figure*}[htb]
    \centering
    \includegraphics[width=0.7\textwidth]{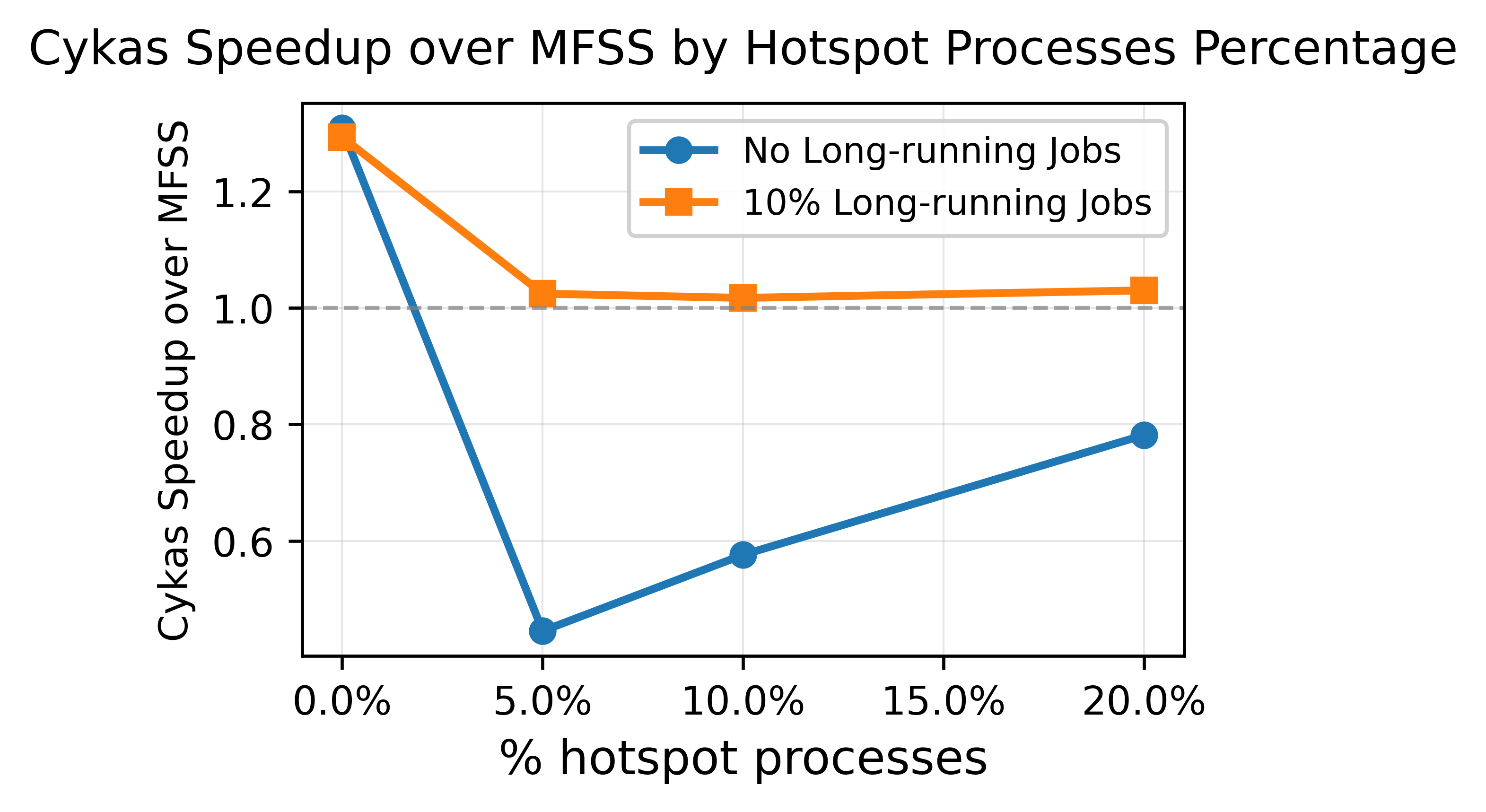}

    \caption{Cykas speedup over MFSS under different percentages of hotspot processes. The blue line shows scenarios without long-running jobs, while the orange line shows scenarios where 10\% of application messages trigger long-running jobs. Values > 1.0 indicate Cykas outperforms MFSS.}
    \label{fig:hotspot_impact}
\end{figure*}

In~\Cref{subsec:buffer_issue} we motivated Cykas with a scenario where a long-running job is triggered by receiving a message. To understand how Cykas performs relative to MFSS, we examine how job length and communication frequency affect the performance of both protocols.

We compare Cykas and MFSS's performance under varying job lengths (0.5, 5, 12.5, 25, and 50 ms). Our experimental setup uses 100 processes with 50 kBps bandwidth and 5ms propagation delay, resulting in a process-to-bandwidth ratio of 2.0. We set 10\% of application messages to trigger long-running jobs at the recipient when delivered. We generated random workloads in which each process sent 100 messages to a random recipient, where every process had the same probability of receiving a message. Finally, we varied the communication frequency, i.e., how long processes wait between sending messages (1, 10, 100, and 1000 ms).

\Cref{fig:job_impact} shows our results. Each subfigure presents results for a different communication frequency, with job length shown on the x-axis and Cykas speedup over MFSS on the y-axis. Y-axis values above 1.0 indicate Cykas outperforms MFSS. Blue lines show execution time speedup and orange lines show average job start time speedup compared to MFSS.

Across subfigures, we see that Cykas nearly always enjoys an execution time speedup compared to MFSS, and that as job lengths increase to 50ms, the average speedup first increases slightly and then decreases (but remains above 1.0). The intuition for this behavior is that for very short jobs, there is less to be gained from the early job starts that Cykas allows; conversely for very long jobs, execution time becomes dominated by the long-running jobs anyway, and so the effect of starting jobs early is less pronounced.  Between these extremes is a sweet spot for Cykas.  Cykas also enjoys the biggest advantage over MFSS when communication is frequent, as in \Cref{fig:job_impact}~(left). This is because frequent communication provides more opportunities for Cykas to use \EagerSend, whereas when communication is infrequent, Cykas behaves more like MFSS.

The average job starting time speedup follows almost the same trend as the execution time speedup, which  Cykas's faster execution time is most likely brought about by starting jobs earlier.
We conclude that Cykas provides the biggest advantage over MFSS in systems with frequent communication and moderately long-running jobs.

\subsection{Communication Pattern Effects: Uniform vs. Hotspot Distribution}\label{subsec:random_vs_hotspot}

Communication patterns vary across different systems. In this experiment, we focus on two common patterns: uniform and hotspot patterns. In uniform patterns (like the random workloads from the previous section), every participant has roughly the same probability of receiving a message, while in hotspot patterns, a portion of processes (which we call \emph{hotspot processes}) tend to receive more messages than other processes.

We compare MFSS and Cykas execution time under different percentages of hotspot processes in the system. We use the same number of processes, bandwidth, and propagation delay as in~\Cref{subsec:perf_analysis}. The communication frequency is 10ms. The hotspot percentages we test are 0\% (i.e., uniform), 5\%, 10\%, and 20\%. Messages target hotspot processes with 80\% probability.
We consider two scenarios: systems without long-running jobs, and systems where 10\% of messages trigger long-running jobs (mean duration 25ms, normally distributed).

As shown in~\Cref{fig:hotspot_impact}, when there are no hotspot processes (0\% hotspot processes), Cykas outperforms MFSS by 30\% (speedup $\approx$1.3$\times$) in both scenarios. However, as the hotspot percentage increases, Cykas's advantage diminishes. In systems without long-running jobs, Cykas underperforms MFSS when hotspots are present (speedup < 1). Even in systems with long-running jobs, where Cykas maintains its performance advantage across hotspot scenarios, the speedup is lower than the uniform case.

This performance degradation occurs because hotspot processes become communication bottlenecks under Cykas. Frequent \EagerSend messages force these heavily-targeted processes into secret mode repeatedly, blocking subsequent message sending and creating system-wide delays. Uniform distribution avoids this bottleneck effect, allowing Cykas to leverage its eager sending advantages.
We conclude that Cykas is best suited for uniform communication patterns, where message traffic is evenly distributed across processes.

\section{Related Work}\label{sec:related}

Protocols for causal message delivery have been studied since at least the 1980s. Some protocols are specific to the case of broadcast messages~\cite{birman-virtual-synchrony,birman-reliable,birman-lightweight-cbcast}, while others are more general~\cite{raynal1991causal,schiper-causal-ordering,mattern2005non}. Our focus here is on the general setting.
With the exception of the MFSS~\cite{mattern2005non} protocol, receiver-side enforcement of causal delivery has been the norm in the literature, though Birman et al.~\cite{birman-lightweight-cbcast} also briefly mention a sender-side technique for causal delivery of point-to-point messages ``in which the sender is inhibited from starting new multicasts until reception of the point-to-point message is acknowledged''.

The MFSS protocol is stronger in a formal sense than general causally ordered communication --- that is, the executions allowed by it are a strict subset of those allowed by the causal order. Di Giusto et al.~\cite[\S7]{di-giusto-message-passing} hypothesize that \cite{mattern2005non}'s protocol is ``somewhere between mailbox and causally ordered'' in terms of its strength (where ``mailbox'' refers to a communication model that is slightly more restrictive than causally ordered communication~\cite{di-giusto-message-passing}). Receiver-side protocols, on the other hand, allow \emph{all} executions allowed by the causal order. Our aim with Cykas is to recover some of the message ordering flexibility of the receiver-side approach while avoiding the metadata overhead it imposes.

Nieto et al.~\cite{nieto-crdts-aneris} and Redmond et al.~\cite{redmond-cbcast-lh} mechanically verify the safety of executable causal broadcast protocols. Nieto et al. implement a verified-safe causal broadcast protocol as part of a larger proof development that also verifies the correctness of conflict-free replicated data types, using the  Aneris separation logic framework in the Coq proof assistant. Redmond et al. express causal message delivery property as a refinement type in Liquid Haskell, and use it to verify the safety of a Haskell implementation of Birman et al.~\cite{birman-lightweight-cbcast}'s causal broadcast protocol.  These verification efforts focus on broadcast messages rather than point-to-point messages, and both consider only receiver-side protocols. To our knowledge, ours is the first mechanized verification of an implementation of a sender-side causal delivery protocol, and the first to consider a liveness property. Of course, our verification effort is \emph{bounded}, unlike Nieto et al.'s and Redmond et al.'s proof developments.

\section{Conclusion and Future Work}
\label{sec:discussion}

In this paper, we have presented the Cykas protocol, a new sender-side protocol for causal message delivery that eagerly sends messages that would traditionally need to be delayed at the sender's side. We show that Cykas enjoys the scalability of the traditional sender-side approach, but with a performance advantage that is most evident for systems with frequent communication, uniform communication patterns, and moderately long-running jobs.

In Cykas, the decision of whether a message should be sent normally or eagerly is not exposed to the application level; this is by design, so that the protocol is a drop-in replacement for existing causal delivery protocols.  However, it could be useful to expose both message-sending options to the application level; future work could explore such a design. More generally, an issue with any causal delivery protocol is that the happens-before relation is quite coarse-grained and, depending on the application, may relate messages that need not actually be delivered in causal order for application correctness.  We hypothesize that language-level information flow analysis could play a role in determining the sendability of a message and that the Cykas protocol could be a jumping-off point for such an investigation.

\bibliographystyle{ACM-Reference-Format}
\bibliography{refs}

\end{document}